\documentclass[twocolumn, pra , nofootinbib , superscriptaddress
]{revtex4-1}
\usepackage{graphicx}
\usepackage{amsfonts}
\usepackage{amssymb}
\usepackage{amsthm}
\usepackage{array}
\usepackage{amsmath}
\usepackage{verbatim} 
\usepackage{hyperref}
\usepackage{color}
\usepackage{bbold}
\usepackage{epstopdf}
\usepackage{mathtools}
\usepackage{enumerate}
\usepackage{subcaption}
\usepackage{ulem}
\usepackage{physics}
\usepackage{subcaption}
\usepackage{soul,xcolor}

\newtheorem*{proposition*}{Proposition}

\newtheorem{theorem}{Theorem}
\newtheorem*{theorem*}{Theorem}
\newtheorem{corollary}{Corollary}[theorem]
\newtheorem*{corollary*}{Corollary}

\setstcolor{red}

\begin{document}
\title{Identifying quantum phase transitions via geometric measures of nonclassicality}
\author{Kok Chuan Tan}
\email{bbtankc@gmail.com}
\affiliation{ School of Physical and Mathematical Sciences, Nanyang Technological University, Singapore 637371, Republic of Singapore}

\begin{abstract}
In this article, we provide theoretical support for the use of geometric measures of nonclassicality as a general tool to identify quantum phase transitions. We argue that divergences in the susceptibility of any geometric measure of nonclassicality are sufficient conditions to identify phase transitions at arbitrary temperature.  This establishes that geometric measures of nonclassicality, in any quantum resource theory, are generic tools to investigate phase transitions in quantum systems. At zero temperature, we show that geometric measures of quantum coherence are especially useful for identifying first order quantum phase transitions, and can be a particularly robust alternative to other approaches employing measures of quantum correlations.
\end{abstract}

\maketitle

\section{Introduction}

The development of various characterizations and notions of nonclassicality in recent years have lead to several proposals to apply such notions in order to probe a system undergoing a phase transition. Examples include entanglement\cite{Osterloh2002, Osborne2002, Wu2004}, quantum discord\cite{Sarandy2009, Werlang2010}, and more recently, quantum coherence\cite{Karpat2014, Chen2016, Malvezzi2016, Qin2018}. As these methods study the intrinsic nonclassical properties of quantum states, they do not require any prior knowledge about the order parameters associated with the phase transition. Notions of quantum nonclassicality are also often accompanied by novel physical interpretations. For instance, entanglement has operational interpretations in terms of quantum teleportation\cite{Bennett1991}, quantum cryptography\cite{Ekert1991}, and superdense coding\cite{Bennett1992}. Quantum discord has been shown to be a useful resource for entanglement distribution\cite{Chuan2012} and remote state preparation\cite{Dakic2012}. Quantum coherence has been applied to quantum state merging\cite{Streltsov2016}, speed-ups in quantum computation\cite{Hillery2016, Shi2017}, and nonclassical light\cite{Tan2017}. 
Probing phase transitions using such notions of nonclassicality therefore opens up the use of powerful mathematical machinery that was developed in order to study and interpret nonclassicality in the quantum information sciences\cite{Horodecki2001, Modi2012, Streltsov2017, Tan2019, Chitambar2019}. 

In this article, we provide geometric arguments justifying the use of geometric measures of nonclassicality\cite{Wei2003, Dakic2010, Baumgratz2014} in the identification and detection of quantum phase transitions\cite{Wei2005, Cheng2016, Sha2018, Qin2018, Malvezzi2016, Sha2018} under very general conditions. Specifically, we argue that for arbitrary quantum resource theories, divergences in the geometric nonclassical susceptibility or its first derivative are sufficient conditions for identifying phase transitions. This opens up the use of any geometric measure of nonclassicality, including but not limited to entanglement, quantum discord, or quantum coherence, to probe the phase transitions of a quantum system. 

In particular, for quantum phase transitions at zero temperature, we show that first order phase transitions where the system experiences a sudden change in the ground state can always be identified via a diverging geometric coherence susceptibility. This is true even when entanglement or quantum discord may potentially fail to identify the phase transition. We also show how the geometric coherence susceptibility may be a more general approach than many other methods employing Berry phases\cite{Carollo2005, Zhu2006} or order parameters. This suggests that out of all the possible measures of nonclassicality, measures of quantum coherence may be particularly relevant to the study of phase transitions in quantum systems.

\section{Preliminaries}

A phase transition is characterized by dramatic changes in the system of interest when there is a small variation in some physical control parameter $\lambda$. A critical parameter is then some value $\lambda = \lambda_c$ where a phase transition occurs. 

A quantum phase transition (QPT)\cite{Vojta2003} is defined as a phase transition that occurs at zero temperature. At zero temperature, contributions from thermal fluctuations are completely removed from consideration and since thermal fluctuations are typically considered to be classical contributions, any critical phenomena that remains can be thought of as purely quantum in nature. In this scenario, a system in thermal equilibrium occupies the ground state of the Hamiltonian, $H(\lambda)$, which depends on some control parameter $\lambda$.

In this article, we will adopt a geometric approach to the study of phase transitions. Suppose the control parameter $\lambda$ is a real number which labels the points along some curve in state space, $\rho(\lambda)$. We then consider some distance measure, also called a metric, $D$, within this state space. Recall that $D$ is a proper distance measure when it satisfies the following properties: for any quantum states $\rho$, $\sigma$ and $\tau$, (i) $D(\rho, \sigma) \geq 0$ (ii) $D(\rho,\sigma) = 0$ iff $\rho = \sigma$, (iii) $D(\rho,\sigma) = D(\sigma,\rho)$ and (iv) $D(\rho, \sigma) \leq D(\rho,\tau) + D(\sigma, \tau)$. The last property is particularly noteworthy and is called the triangle inequality. 

For a given distance measure $D$, we will consider the distance between two infinitesimally close states along the curve $\rho(\lambda)$ and $\rho(\lambda+\delta\lambda)$. This is called a line element and is denoted $\dd{s}$. Its derivative with respect to $\lambda$ is denoted $\dd{s}/\dd{\lambda} \coloneqq \lim_{\delta \lambda \rightarrow 0^+}D[\rho(\lambda+\delta\lambda), \rho(\lambda)]/\delta\lambda$. Note that this is defined as a limit over positive $\delta \lambda$ so $\dd{s}/\dd{\lambda}$ is a non-negative quantity that directly quantifies the rate of change occurring in the system for an infinitesimal variation in the control parameter. When a system undergoes a phase transition, it is expected that $\dd{s}/\dd{\lambda}$ becomes non-analytic, as that signals structural changes in the system when $\lambda$ is varied. This viewpoint is in line with the differential geometric approach, which identifies quantum phase transitions via non-analyticities in the quantum geometric tensor\cite{Zanardi2007}. We note that given a (Riemannian) metric tensor, a proper distance measure may be defined, while the converse may not be true in general. In this sense, the derivative $\dd{s}/\dd{\lambda}$ can be considered a generalization of the quantum geometric tensor approach.

A quantum coherence measure is a basis dependent measure of the amount of quantum superposition amongst orthogonal quantum states. Let us consider a complete basis $\{ \ket{e_i} \}$. We say that a quantum state $\rho$ is incoherent if its density matrix has no non-zero off-diagonal elements, i.e. $\rho_{ij} \coloneqq \bra{e_i}\rho \ket{e_j} = 0$ for every $i \neq j$. Otherwise, we say that the state has coherence.  Since the diagonal elements of the density matrix and hence the coherence is always defined with respect to some given basis $\{ \ket{e_i} \}$, this is called the incoherent basis, and a state $\rho$ is incoherent if and only if its density matrix is diagonal with respect to this basis.

An important class of coherence measures are the so-called geometric coherence measures. The geometric coherence is defined as the quantity
\begin{align*}
\mathcal{C}_D(\rho) \coloneqq \min_{\sigma \in \mathcal{I}} D(\rho,\sigma),
\end{align*} where $D$ is some distance measure and the minimization is over the set of all incoherent states $\mathcal{I}$. For instance, one can choose the distance measure $D(\rho, \sigma)$ to be $\norm{\rho - \sigma}_{l1}$ where $\norm{\cdot}_{l1}$ is the $l1$-norm. This then gives rise to the so-called $l1$-norm of coherence\cite{Baumgratz2014} which turns out to be the absolute sum of all off diagonal elements $\mathcal{C}_{l1}(\rho) = \sum_{i\neq j} \abs{\rho_{ij}} .$

Based on the above definition of the geometric coherence, we can also define the geometric coherence susceptibility (GCS), which quantifies the rate of change of the geometric coherence of a state $\rho(\lambda)$ with respect to a change in the parameter $\lambda$. It is defined as
\begin{align*}
\mathcal{X}_D[\rho(\lambda)] &\coloneqq \dv{\mathcal{C}_D[\rho(\lambda)]}{\lambda} \\
&= \lim_{\delta \lambda \rightarrow 0} \{\mathcal{C}_D[\rho(\lambda+\delta\lambda)] - \mathcal{C}_D[\rho(\lambda)] \} / \delta \lambda. 
\end{align*}

More generally, for arbitrary quantum resource theories, one can also similarly define a geometric nonclassicality quantifier $\mathcal{N}_D(\rho) \coloneqq \min_{\sigma \in \mathcal{S}} D(\rho,\sigma)$ where $\mathcal{S}$ is any set of classical states.  The corresponding geometric nonclassical susceptibility (GNS) is denoted $\mathcal{X_{\mathcal{N},D}}[\rho(\lambda)]$. Such measures play a significant role in the study of quantum resources such as entanglement\cite{Wei2003} and quantum discord\cite{Dakic2010}.

In the following sections, we will consider the role of GNS and GCS in identifying QPTs.

\section{Geometric nonclassical susceptibility and phase transitions at arbitrary temperature}

We first consider the GNS for arbitrary quantum resource theories. It may be expected that a diverging GNS implies a sudden structural change in the system and therefore indicates a phase transition. The following theorem provides a general geometric argument that is true for any quantum resource under consideration. 

\begin{theorem} \label{thm::sufficientcond}
	If the nonclassical susceptibility $\mathcal{X_{\mathcal{N},D}}[\rho(\lambda)]$  diverges at some critical parameter $\lambda = \lambda_c$, then $\dd{s}/\dd{\lambda}$ also diverges and $\lambda_c$ is a critical parameter indicating a phase transition
\end{theorem}

\begin{proof}

Suppose $\rho(\lambda)$ is the ground state density matrix of some Hamiltonian $H(\lambda)$ which depends on some external parameter $\lambda$. The coherence susceptibility of the geometric coherence, w.r.t. some external parameter $\lambda$, is then defined as the quantity
\begin{align*}
\mathcal{X}_{\mathcal{N},D}[\rho(\lambda)] &\coloneqq \dv{\mathcal{N}_D[\rho(\lambda)]}{\lambda} \\
&= \lim_{\delta \lambda \rightarrow 0} \{\mathcal{N}_D[\rho(\lambda+\delta\lambda)] - \mathcal{N}_D[\rho(\lambda)] \} / \delta \lambda. 
\end{align*}

Let us consider $\mathcal{N}_D[\rho(\lambda+\delta\lambda)] - \mathcal{N}_D[\rho(\lambda)]$.  Without any loss in generality, we can assume that $\mathcal{N}_D[\rho(\lambda+\delta\lambda)] - \mathcal{N}_D[\rho(\lambda)] \geq 0$ as otherwise we can always reparametrize $\lambda$ to go in the other direction such that the assumption will always be true. Suppose $\sigma(\lambda)$ is the optimal state that achieves $\mathcal{N}_D[\rho(\lambda)] = D[\rho(\lambda),\sigma(\lambda)]$. We have the following series of inequalities: 
\begin{align}
\mathcal{N}_D & [\rho(\lambda+\delta\lambda)] - \mathcal{N}_D[\rho(\lambda)] \\
&= \min_{\sigma \in \mathcal{I}} D[\rho(\lambda+\delta\lambda),\sigma] - \min_{\sigma \in \mathcal{I}} D[\rho(\lambda),\sigma] \label{eq:susc1}\\
&\leq D[\rho(\lambda+\delta\lambda),\sigma(\lambda)] - D[\rho(\lambda),\sigma(\lambda)] \label{eq:susc2}\\
&\leq D[\rho(\lambda+\delta\lambda), \rho(\lambda)]. \label{eq:susc3}
\end{align}
	
Eq.~(\ref{eq:susc1}) comes from the definition of geometric coherence. 	The inequality in Eq.~(\ref{eq:susc2}) comes from fact that $\sigma(\lambda)$ is the optimal state that minimizes the distance to $\rho(\lambda)$, but in general may be suboptimal for the state $\rho(\lambda+\delta\lambda)$. The inequality in Eq.~(\ref{eq:susc3}) comes from the reverse triangle inequality $\abs{D(A,C)-D(B,C)} \leq D(A,B)$. 

Now suppose that $\mathcal{X}_{\mathcal{N},D}$ diverges at point $\lambda_c$ such that $\mathcal{X}_{\mathcal{N},D}[\rho(\lambda_c)] = \infty$, then from the inequality in Eqs.\ref{eq:susc3}, we must also have $D[\rho(\lambda+\delta\lambda), \rho(\lambda)] /\delta\lambda \rightarrow \infty $ as $\delta\lambda \rightarrow 0$. This implies that $\dd{s}/\dd{\lambda}$ diverges at $\lambda = \lambda_c$, so there must be a phase transition at that point.
\end{proof}

It is also frequently observed that instead of a divergence, phase transitions are accompanied by a cusp or a kink in the GNS, i.e. the first derivative of the GNS diverges instead of the GNS itself. The following theorem provides a geometric argument that a cusp or a kink in the GNS can also be used to identify phase transitions under general conditions.

\begin{theorem} \label{thm::2ndOrdSuffCond}
If the first derivative of the nonclassical susceptibility $\dd{\mathcal{X_{\mathcal{N},D}}[\rho(\lambda)]}/\dd{\lambda}$ diverges at some critical parameter $\lambda = \lambda_c$, then $\dd[2]{s}/\dd{\lambda}^2$ diverges and $\dd{s}/\dd{\lambda}$ is non-analytic at $\lambda = \lambda_c$.
\end{theorem}
\begin{proof}
The derivative $\dd{\mathcal{X}_{\mathcal{N},D}[\rho(\lambda)]}/\dd{\lambda}$ is just the second order derivative of the nonclassicality quantifier $\mathcal{N}_{D}[\rho(\lambda)]$. This can be written as the limit
\begin{align*}
&\dv{\mathcal{N}_D[\rho(\lambda)]}{\lambda} 
\\&\coloneqq \lim_{\delta \lambda \rightarrow 0} \frac{\mathcal{N}_{D}[\rho(\lambda+\delta \lambda)]-2\mathcal{N}_{D}[\rho(\lambda)]+\mathcal{N}_{D}[\rho(\lambda-\delta\lambda)]}{\delta \lambda^2}.
\end{align*}

As in the proof of Theorem~\ref{thm::sufficientcond}, let $\sigma(\lambda)$ is the optimal state that achieves $\mathcal{N}_D[\rho(\lambda)] = D[\rho(\lambda),\sigma(\lambda)]$.  We will again assume without any loss in generality that $\dd{\mathcal{X}_{\mathcal{N},D}[\rho(\lambda)]}/\dd{\lambda} \geq 0 $, as otherwise we can just appropriately reparametrize $\lambda$. We then consider the numerator, which can be shown to obey the following series of inequalities:
\begin{align}
&\mathcal{N}_{D}[\rho(\lambda+\delta \lambda)]-2\mathcal{N}_{D}[\rho(\lambda)]+\mathcal{N}_{D}[\rho(\lambda-\delta\lambda)]  \\
&\leq D[\rho(\lambda+\delta \lambda),\sigma(\lambda)] - 2D[\rho(\lambda),\sigma(\lambda)] + D[\rho(\lambda+\delta \lambda),\sigma(\lambda)] \label{2ndOrdSusc1}\\
&\leq D[\rho(\lambda+\delta \lambda),\rho(\lambda)]+ D[\rho(\lambda -\delta \lambda),\sigma(\lambda)] \label{2ndOrdSusc2}\\
&= D[\rho(\lambda+\delta \lambda),\rho(\lambda)]- 2D[\rho(\lambda),\rho(\lambda)]+ D[\rho(\lambda -\delta \lambda),\sigma(\lambda)]. \label{eq::2ndOrderUpperBound}
\end{align} In Eq.~\ref{2ndOrdSusc1}, we used the fact that $\sigma(\lambda)$ is the optimal state that minimizes the distance to $\rho(\lambda)$, but is suboptimal in general for $\rho(\lambda \pm \delta \lambda)$. In Eq.~\ref{2ndOrdSusc2}, we applied the reverse triangle inequality $\abs{D(A,C)-D(B,C)} \leq D(A,B)$. Eq.~\ref{eq::2ndOrderUpperBound} then some from the fact that $D[\rho(\lambda),\rho(\lambda)]=0$, which is a fundamental property of any distance measure $D$.

We then observe that
\begin{align}
\dv[2]{s}{\lambda} &=\lim_{\delta \lambda \rightarrow 0} \{ D[\rho(\lambda+\delta \lambda),\rho(\lambda)] \\
&- 2D[\rho(\lambda),\rho(\lambda)]+ D[\rho(\lambda -\delta \lambda),\sigma(\lambda)] \} / \delta\lambda^2. \label{eq::2ndOrderLimit}
\end{align}

Therefore, if $\dd{\mathcal{X}_{\mathcal{N},D}[\rho(\lambda)]}/\dd{\lambda} = \infty$ at some $\lambda = \lambda_c$, then from Eq.~\ref{eq::2ndOrderUpperBound} and Eq.~\ref{eq::2ndOrderLimit}, we must also have $\dd[2]{s}/\dd{\lambda}^2 = \infty $. Since the derivative of $\dd{s}/\dd{\lambda} $ diverges, it is non-analytic at $\lambda = \lambda_c$.
\end{proof}

Theorems~\ref{thm::sufficientcond} and \ref{thm::2ndOrdSuffCond} provides geometric justification for the use of geometric measures of any quantum resource for identifying quantum phase transitions. We note that in these arguments, no prior assumptions are made about the nature of the state $\rho(\lambda)$, so the results apply to quantum systems at arbitrary temperature. In the following section, prove stronger statements in the zero temperature case, which suggests that geometric coherence measures may be an especially robust tool for identifying quantum phase transitions.

\section{Geometric coherence susceptibility and quantum phase transitions}

In this section, we will consider QPTs at zero temperature. We are therefore interested in probing phase transitions that occur in the ground state $\ket{\psi_0(\lambda)}$ of some Hamiltonian $H(\lambda)$. It is expected that for first order QPTs, a sudden change in the ground state represented by a discontinuity in $\ket{\psi_0(\lambda)}$ across the critical parameter will occur. The following theorem demonstrates that any such change in the ground state is equivalent to the existence of some incoherent basis where GCS diverges.

\begin{theorem} \label{thm::neccSuffCond}
	At zero temperature, a first order quantum phase transition occurs and $\dd{s}/\dd{\lambda}$ diverges at some critical parameter $\lambda = \lambda_c$, if and only if there exists an incoherent basis where the coherence measure $\mathcal{C}_D[\rho(\lambda)]$ is discontinuous at $\lambda = \lambda_c$, and $\mathcal{X}_D[\rho(\lambda)]$ diverges at $\lambda = \lambda_c$.
\end{theorem}

\begin{proof}
	Let the Hamiltonian describing the system be $H(\lambda)$, and the ground state be $\ket{\psi_0(\lambda)}$. The corresponding density matrix is denoted $\rho_0(\lambda)$.

	Suppose for a given distance measure $D$, $\dd{s}/\dd{\lambda}$ diverges at $\lambda=\lambda_c$ and there is a discontinuity in the quantum state along the curve parametrized by $\lambda$. This means that the states as you approach $\lambda=\lambda_c$ from above and below are different, i.e. $\lim_{\delta \rightarrow 0^+} \rho_0(\lambda_c-\delta)  \neq \lim_{\delta \rightarrow 0^+} \rho_0(\lambda_c+\delta)$.
	
	
	Let us choose an incoherent basis $\{ \ket{e_i} \}$ such that $\ket{e_0} = \ket{\psi_0(\lambda_c -\delta)}$, where $\delta > 0 $. We observe that in the basis $\{ \ket{e_i} \}$, $\mathcal{C}_D[\rho_0(\lambda_c-\delta)] = 0$.
	
	Consider $\ket{\psi_0(\lambda_c + \delta)}$ where $\delta >0$. Since $\ket{\psi_0(\lambda_c + \delta)} \neq \ket{\psi_0(\lambda_c - \delta)} $ as we take the limit $\delta \rightarrow 0$, there are only two special cases we need to consider. $\ket{\psi_0(\lambda_c + \delta)}$ is either orthogonal to $\ket{\psi_0(\lambda_c - \delta)}$, or it has partial overlap with $\ket{\psi_0(\lambda_c - \delta)}$.
	
	If it is orthogonal, we can just choose a basis where $\ket{e_i} \neq \ket{\psi_0(\lambda_c + \delta)}$ for every $i \geq 1$. Since $\ket{\psi_0(\lambda_c + \delta)}$ is not an element of the incoherent basis, this means that $\mathcal{C}_D[\rho_0(\lambda_c+\delta)] > 0$ even in the limit $\delta \rightarrow 0$.
	
	If there is partial overlap, then we can write $\ket{\psi_0(\lambda_c + \delta)} = a \ket{\psi_0(\lambda_c - \delta)} + b \ket{\psi^\perp}$, where $\ket{\psi^\perp}$ is some normalized vector orthogonal to $\ket{\psi_0(\lambda_c - \delta)}$. Since there is a discontinuity in the ground state, we are guaranteed that $b$ will not go to zero as $\delta \rightarrow 0$. We can therefore choose $\ket{e_1} = \ket{\psi^\perp}$.  Since $\ket{e_0} = \ket{\psi_0(\lambda_c -\delta)}$ and $\ket{\psi_0(\lambda_c + \delta)} = a \ket{e_0} + b \ket{e_1}$, this means that we have $\mathcal{C}_D[\rho_0(\lambda_c+\delta)] > 0$ even in the limit $\delta \rightarrow 0$.
	
	In either case, it suggests that we can always find a basis $\{ \ket{e_i} \}$ where $\lim_{\delta \rightarrow 0^+} \mathcal{C}_D[\rho_0(\lambda_c-\delta)] = 0$ and  $\lim_{\delta \rightarrow 0^+}\mathcal{C}_D[\rho_0(\lambda_c+\delta)] > 0$, so  $\mathcal{C}_D[\rho_0(\lambda)]$ is a step function in the immediate vicinity of $\lambda = \lambda_c$. This implies $\mathcal{X}_D[\rho(\lambda)]$ diverges at $\lambda = \lambda_c$. This proves the theorem in the forward direction.
	
	For the converse direction, suppose the coherence measure $\mathcal{C}_D$ is discontinuous and $\lim_{\delta \lambda \rightarrow 0^+} \abs{ \mathcal{C}_D[\rho(\lambda_c - \delta \lambda)] -  \mathcal{C}_D[\rho(\lambda_c + \delta \lambda)]} = \Delta $, for some $\Delta > 0$. Without any loss in generality, we will assume that the coherence decreases as we increase $\lambda$ such that $\mathcal{C}_D[\rho(\lambda_c - \delta \lambda)] >  \mathcal{C}_D[\rho(\lambda_c + \delta \lambda)]$, as otherwise we can reparametrize $\lambda$ to go in the other direction. Let $\sigma$ be the optimal state achieving $\mathcal{C}_D[\rho(\lambda_c + \delta \lambda)] = D[\rho(\lambda_c+\delta \lambda),\sigma]$. We then have the following series of inequalities:
	
	\begin{align}
	& \mathcal{C}_D[\rho(\lambda_c - \delta \lambda)] -  \mathcal{C}_D[\rho(\lambda_c + \delta \lambda)] \\ 
	&\leq D[\rho(\lambda_c - \delta \lambda), \sigma] - D[\rho(\lambda_c + \delta \lambda), \sigma] \label{eq::nessSuff1}\\
	&\leq D[\rho(\lambda_c - \delta \lambda), \rho(\lambda_c + \delta \lambda)] \label{eq::nessSuff2}
	\end{align}
	
	In Eq.~\ref{eq::nessSuff1}, we used the definition $\mathcal{C}_D(\rho) = \min_{\sigma \in \mathcal{I}} D[\rho,\sigma]$ and the fact that $\sigma$ is optimal for the state $\rho(\lambda_c+\delta\lambda)$, but is in general suboptimal for $\rho(\lambda_c-\delta\lambda)$. In Eq.~\ref{eq::nessSuff2}, we used the inverse triangle inequality $\abs{D(A,C)-D(B,C)} \leq D(A,B)$.
	
	Finally, combining Eq.~\ref{eq::nessSuff2} and the fact that  $\lim_{\delta \lambda \rightarrow 0^+} \abs{ \mathcal{C}_D[\rho(\lambda _c- \delta \lambda)] -  \mathcal{C}_D[\rho(\lambda_c + \delta \lambda)]} = \Delta $ implies $\lim_{\delta \lambda \rightarrow 0^+} D[\rho(\lambda_c - \delta \lambda), \rho(\lambda_c + \delta \lambda)]  > \Delta$, which shows that there is a discontinuity in the ground state, so there is a first order QPT.
	
\end{proof}

Theorem~\ref{thm::neccSuffCond} therefore singles out geometric measures of quantum coherence as a useful tool to probe first order QPTs where other nonclassical measures may potentially fail. We will illustrate this with an example in a subsequent section.

\section{Geometric coherence susceptibility, Berry phases, and order parameters}

For many systems, the Berry phase is a useful tool for studying QPTs. In this section, we consider how the GCS is related to the Berry phase at a critical parameter.

Suppose the ground state of the Hamiltonian $H(\lambda)$ is $\ket{\psi_0(\lambda)}$ and that the system is adiabatically evolved through some close looped trajectory in state space. In such a case, the evolution of the ground state at any point along this closed loop can be described by $U(\mu)\ket{\psi_0(\lambda)}$, where $U(\mu) \coloneqq e^{-iG(\mu)}$ and $G(\mu)$ is a Hermitian operator that depends on the parameter $\mu \in [0,2\pi]$. Since the trajectory follows a closed loop, the unitary $U(\mu)$ and the Hermitian operator $G(\mu)$ must satisfy the cyclic property $U(0)\ket{\psi_0(\lambda)} = \ket{\psi_0(\lambda)} = U(2\pi)\ket{\psi_0(\lambda)}$. 

We now consider the Berry phase generated by an evolution described by $G(\mu) = \mu O$, where $O$ is some Hermitian operator.

\begin{corollary} \label{thm::BerryPhaseCoh}
	Consider a Berry phase generated by a cyclic unitary of the type $U(\mu) \coloneqq e^{-i \mu O}$, $\mu \in [0,2\pi]$ acting on a ground state $\ket{\psi_0}$ of the system Hamiltonian $H(\lambda)$.
	
	Suppose at some critical parameter $\lambda = \lambda_c$  that the Berry phases just before and after the critical parameter is given by $\phi^B(\lambda^-_c) = \phi^-$ and $\phi^B(\lambda^+_c) = \phi^+ $ respectively,  where $\lambda^-_c < \lambda_c < \lambda^+_c$. 
	
	Then the Berry phase is discontinuous such that $\phi^- \neq \phi^+$ only if $\mathcal{X}_{D}(\ket{\psi_0(\lambda)})$ is divergent at $\lambda = \lambda_c$ for some incoherent basis.
	
\end{corollary}

\begin{proof}
	
	We first compute the Berry phase generated by the unitary $U(\mu)$. It can be verified that it is given by 
	\begin{align}
	\phi^B &= i \int_0^{2\pi}\bra{\psi_0(\lambda)}U(\mu)^\dagger \dv{\mu} U(\mu) \ket{\psi_0(\lambda)} \\
	&= -i^2 \int_0^{2\pi}\bra{\psi_0(\lambda)}U(\mu)^\dagger U(\mu) O \ket{\psi_0(\lambda)} \\
	&=2 \pi  \bra{\psi_0(\lambda)} O \ket{\psi_0(\lambda)}, 
	\end{align} where $\ket{\psi_0(\lambda)}$ is the ground state of the Hamiltonian $H(\lambda)$. The density matrix of $\ket{\psi_0(\lambda)}$ is denoted $\rho_0(\lambda) = \ket{\psi_0(\lambda)}\bra{\psi_0(\lambda)}$. 
	
	Suppose the Berry phase is discontinuous and $\phi^- \neq \phi^+$. This implies that $\phi^-/2\pi  = \bra{\psi_0(\lambda^-_c)} O \ket{\psi_0(\lambda^-_c)} \neq \bra{\psi_0(\lambda^+_c)} O \ket{\psi_0(\lambda^+_c)} = \phi^+/2\pi$ for $\lambda^-_c < \lambda_c < \lambda^+_c$.  This implies $\lim_{\delta \rightarrow 0^+} \rho_0(\lambda_c+\delta/2)  \neq \lim_{\delta \rightarrow 0^+} \rho_0(\lambda_c-\delta/2)$, so there is a first order QPT at $\lambda=\lambda_c$. Theorem~\ref{thm::neccSuffCond} then shows that the coherence susceptibility $\mathcal{X}_{D}(\ket{\psi_0(\lambda)})$ is divergent for some incoherent basis at the critical parameter.

\end{proof}

A similar argument also shows that if a first order QPT is identifiable by some order parameter, then it must also be identifiable by a diverging GCS.

\begin{corollary} \label{thm::OrderParam}
	Let $O$ be some order parameter for a system described by a Hamiltonian $H(\lambda)$. Let $\ket{\psi_0(\lambda)}$ be the ground state of the system. 
	
	Suppose at some critical parameter $\lambda = \lambda_c$, the mean value of the order parameter as we approach the critical parameter from below and above are $\bra{\psi_0(\lambda_c^-)}O\ket{\psi_0(\lambda_c^-)} $  and $\bra{\psi_0(\lambda_c^+)}O\ket{\psi_0(\lambda_c^+)} $ respectively, where $\lambda_c^- < \lambda_c < \lambda_c^+$.
	
	Then the mean value of the order parameter is discontinuous and $\bra{\psi_0(\lambda_c^-)}O\ket{\psi_0(\lambda_c^-)}  \neq \bra{\psi_0(\lambda_c^+)}O\ket{\psi_0(\lambda_c^+)}$ only if $\mathcal{X}_{D}(\ket{\psi_0(\lambda)})$ is divergent at $\lambda = \lambda_c$ for some incoherent basis.
	
\end{corollary}

\begin{proof}
	Let the density matrix of the ground state $\ket{\psi_0(\lambda)}$ be denoted by $\rho_0(\lambda) = \ket{\psi_0(\lambda)}\bra{\psi_0(\lambda)}$. 
	
	Since $\bra{\psi_0(\lambda_c^-)}O\ket{\psi_0(\lambda_c^-)}  \neq \bra{\psi_0(\lambda_c^+)}O\ket{\psi_0(\lambda_c^+)}$ for $\lambda_c^- < \lambda_c < \lambda_c^+$, and this remains true even as $\lambda_c^-$ and $\lambda_c^+$ approaches $\lambda_c$, we must have that $\lim_{\delta \rightarrow 0} \rho_0(\lambda_c+\delta/2)  \neq \lim_{\delta \rightarrow 0} \rho_0(\lambda_c-\delta/2)$. The rest of the argument follows identically as Corollary~\ref{thm::BerryPhaseCoh}.
\end{proof}

Corollaries~\ref{thm::BerryPhaseCoh} and~\ref{thm::OrderParam} demonstrate how the GCS approach is more general than many types of Berry phases or order parameters that are used to identify first order phase transitions. This supports the view that GCS can be a more general alternative for probing QPTs.

\section{Example}

We consider a  one dimensional spin-$\frac{1}{2}$ chain with $XY$ interaction. The simplest example of this is a two spin system. As we shall see, this example is particular instructive, and describes many of the salient features of the results that were discussed. The Hamiltonian is given by 
\begin{align*}
H(\delta, h) = - \frac{1+\delta}{2} \sigma^x_1 \sigma^x_2 - \frac{1-\delta}{2} \sigma^y_1\sigma^y_2 -\frac{h}{2}(\sigma^z_1 + \sigma^z_2  ).
\end{align*} The parameter $\delta$ describes the anisotropy between the $X$ and the $Y$ interactions, while $h$ describes the strength of the local magnetic field. For any given $\delta$ and $h$, let $\ket{\psi_0(\delta,h)}$ denote the ground state of $H(\delta,h)$.

The above system is described by a $4\times4$ matrix, so we can directly compute the eigenvalues and eigenvectors. One may verify that the Hamiltonian has the eigenvalues $\pm 1$ and $\pm r$, where $r = \sqrt{\delta^2 + h^2}$. The eigenvector corresponding to the eigenvalue $-1$ is the odd parity state $\ket{g_{-}} \coloneqq (\ket{01}_{1,2} +\ket{10}_{1,2})/\sqrt{2}$. The eigenvector corresponding to the eigenvalue $-r$ is the even parity state $\ket{g_{+} (\delta,h)} \coloneqq \cos \frac{\theta}{2} \ket{00}_{1,2} +\sin \frac{\theta}{2}\ket{11}_{1,2},$ where $\tan \theta \coloneqq \delta/h$.

We see that the ground state of the system depends on the value of $r$. When $r > 1$,  the ground state is $\ket{\psi_0(\delta,h)} = \ket{g_+ (\delta,h)}$. When $r < 1$, the ground state is $\ket{\psi_0(\delta,h)} = \ket{g_-}$. The point $r=1$ therefore identifies a critical parameter, since there is a sudden change in the ground state around this point.

Furthermore, let us consider the Berry phase generated by the cyclic unitary $U(\mu) = \exp[-i \mu (\sigma^z_1 +\sigma^z_2)/2]$. Such Berry phases have been experimentally observed in Ref.~\cite{Peng2010}. One may verify that this will transform the ground state such that 
\begin{align*}
&U(\mu)\ket{g_-} = \ket{g_-} = \frac{1}{\sqrt{2}}(\ket{01}_{1,2} +\ket{10}_{1,2})\\
&U(\mu)\ket{g_+(\delta,h)} =  \cos\frac{\theta}{2} \ket{00}_{1,2} +e^{-2i\mu} \sin \frac{\theta}{2}\ket{11}_{1,2}.
\end{align*} Integrating over $\mu \in [0,2\pi]$, we observe that when $r<1$, the Berry phase is $\phi^- = 0 $, and when $r>1$, the accumulated Berry phase is $\phi^+ = -4\pi \cos \theta$. 

Finally, we can choose the total magnetization $O = \sigma^z_1+\sigma^z_2$ to be the order parameter. We see that when $r < 1$, $\bra{g_-} O \ket{g_-} = 0$, and when $r >1$, we have $\bra{g_+(\delta,h)} O \ket{g_+(\delta,h)}= \sin^2 \frac{\theta}{2} - \cos^2 \frac{\theta}{2} = -\cos \theta$.

 The sudden change in the ground state, in conjunction with the sudden accumulation of the Berry phase and the change in the magnetization when $0<\abs{\cos \theta} <1$, suggests that the QPT may be detected by observing the divergences in the GCS (see Theorems~\ref{thm::sufficientcond} and Corollaries-\ref{thm::BerryPhaseCoh} and~\ref{thm::OrderParam}). 
 
 To verify this, let us choose $D$ to be the $l1$-norm induced distance and $\mathcal{C}_{l1}$ to be the $l1$-norm of coherence. We then compute the $l1$-norm of coherence in the computational basis $\{ \ket{00}_{1,2}, \ket{01}_{1,2}, \ket{10}_{1,2}, \ket{11}_{1,2} \}$. We see that $\mathcal{C}_{l1}(\ket{g_-}) = 1$, while $\mathcal{C}_{l1}[\ket{g_+(\delta, h)}] = \abs{ \sin \theta} < 1$ when  $0<\abs{\cos \theta} <1$. The coherence is therefore a step function in the vicinity of $r = 1$, which means the coherence susceptibility $\mathcal{X}_{l1} [ \ket{\psi_0(\delta, h)} ] $ diverges at $r = 1$.
 
 Let us now consider the special case where $h = 0$. We then have $\theta = \pi/2$ and there is a quantum phase transition occurring at $\delta = 1$. In this case, $\ket{g_+(\delta, h = 0)} = \frac{1}{\sqrt{2}} ( \ket{00}_{1,2} + \ket{11}_{1,2})$, which is a maximally entangled state. Observe that $\ket{g_-}$ is also maximally entangled. The total entanglement in the system therefore does not change at the energy level crossing $\delta = 1$. This transition is therefore not detected by divergences in the entanglement susceptibility. Note that since quantum discord and entanglement are equivalent over the set of pure states, discord measures will also not be able to detect this energy level crossing. 
 
  Theorem~\ref{thm::sufficientcond} however, suggests that we should be able to find an incoherent basis where the GCS diverges. Indeed, one can compute $\mathcal{C}_{l1}$ in the incoherent basis $\{ \frac{1}{\sqrt{2}} ( \ket{01}_{1,2} + \ket{10}_{1,2}), \frac{1}{\sqrt{2}} ( \ket{01}_{1,2} - \ket{10}_{1,2}),  \ket{00}_{1,2} , \ket{11}_{1,2} \}$. In this basis, we see that $\mathcal{C}_{l1}(\ket{g_-}) = 0$ and $\mathcal{C}_{l1}[\ket{g_+(\delta, h)}] = 1$, so again, the coherence is a step function in the vicinity of the critical parameter and $\mathcal{X}_{l1} [ \ket{\psi_0(\delta, h)} ] $ diverges at $\delta = 1$.

For the more general case of $N$ spins, the Hamiltonian has the form 
\begin{align*}
H(\delta, h) = -\sum^N_{j=1} \left( \frac{1+\delta}{4} \sigma^x_j \sigma^x_{j+1} + \frac{1-\delta}{4} \sigma^y_j\sigma^y_{j+1} +\frac{h}{2}\sigma^z_j \right ) ,
\end{align*} where $N$ is the total number of spins, $\sigma^a_j$, $a=x,y,z$ are the canonical Pauli operators acting on the $j$th spin.  In Ref~\cite{Wei2005}, the derivative of the entanglement density was investigated in the thermodynamic limit $N \rightarrow \infty$. It was observed that the phase transition at $h=0$, $\delta = 1$ was not identified by the entanglement susceptibility.


It is known that for $N>1$, the ground state belongs to either one of the parity sectors. At $h = \pm \sqrt{1-\delta^2}$, i.e. $r= \sqrt{\delta^2 + h^2 }= 1$, a phase transition occurs where the parity of ground state flips\cite{Pasquale2009}. We already see this from the two spin case, where we see that at $r = 1$, there is an energy level crossing and the ground state flips from the odd parity state $\ket{g_-}$ to the even parity state $\ket{g_+(\delta, h)}$. 

For $r < 1$, let the ground state be $\ket{g_p(\delta, h)}$, and for $r>1$, let the ground state be $\ket{g_{-p}(\delta, h)}$, where $p = \pm $ denotes the parity of the ground state. 

Suppose the subspace with parity $-p$ is spanned by some orthonormal set $\{ \ket{e_m} \}_{m=0}^{M_p-1}$ and the subspace with parity $p$ is spanned by another orthonormal set $\{ \ket{e_m} \}_{m=M_p}^{2^N}$. For the subspace with parity $p$, we can choose $\ket{e_{M_p}} = \ket{g_p(\delta^-, h^-)}$ at some $\delta^-, h^-$ very close to criticality such that $r <1$. For the subspace with parity $-p$, we choose $\ket{e_0} = \ket{g_{-p}(\delta^+, h^+)}$ at some $\delta^+,h^+$ very close to criticality such that $r>1$. We then perform a Fourier transform $\ket{e'_{m'}} \coloneqq \sum_{m=0}^{M_p-1} \exp[-\frac{2\pi i}{M_p}mm']/\sqrt{M_p}\ket{e_m}$. Finally, we can choose our incoherent basis to be $\{ \ket{e'_m} \}_{m=0}^{M_p-1} \bigcup \{ \ket{e_m} \}_{m=M_p}^{2^N}$.

The above prescription ensures that $\mathcal{C}_D(\ket{g_p(\delta^-,h^-}) = 0$ since $ \ket{g_{p}(\delta^-, h^-)}$ belongs to the incoherent basis, but $\mathcal{C}_D(\ket{g_p(\delta^+,h^+}) > 0$ since $\ket{g_p(\delta^+,h^+)}$ is not an element of the incoherent basis. There is therefore a sudden jump in the coherence as we cross $r=1$, so we are guaranteed that $\mathcal{X}_D ( \ket{ \psi_0 (\delta, h)} ) $ diverges at the critical parameter for arbitrary $N$. 

See Refs.~\cite{Qin2018, Malvezzi2016, Sha2018} for further examples where geometric measures of coherence were also used to identify QPTs.

\section{conclusion}

In this article, we considered the role that geometric measures of nonclassicality play in the identification of phase transitions. Theorems~\ref{thm::sufficientcond} and \ref{thm::2ndOrdSuffCond} show that geometry based measures of nonclassicality are generic tools that can be used to probe phase transitions at arbitrary temperature. These results apply to any quantum resource theory, which include notions such as entanglement, quantum discord and quantum coherence. While we have only considered geometric measures of nonclassicality, one may also expect that many non-geometric measures will exhibit similar behaviour during phase transitions. This is because both geometric and non-geometric measures are ultimately trying to capture the same underlying notion of nonclassicality.    

We then considered QPTs at zero temperature. In this regime, we showed in Theorem~\ref{thm::neccSuffCond} that any sudden change in the ground state at the point of criticality can always be picked up by a diverging GCS, measured with respect to some incoherent basis. In support of this, Theorem~\ref{thm::BerryPhaseCoh} and Corollary~\ref{thm::OrderParam} show that large classes of QPTs that can be detected via Berry phases or order parameters can also be detected by a diverging GCS. 

We illustrate the case by considering a toy model consisting of 2 qubits with XY interaction. We show that an energy level crossing in this model cannot be detected using entanglement or coherence measures, since the total quantum correlation remains unchanged. By appropriately defining an incoherent basis however, one can demonstrate a diverging GCS at the point of phase transition. This points to the utility of quantum coherence measures as an alternative for probing certain types of QPTs where other quantum correlations based methods may fail. We can intuitively understand this to be because quantum correlations such as entanglement and discord may be viewed as special kinds of quantum coherence\cite{Tan2016, Tan2018}.

We hope that this work will spur continued research on the relationship between nonclassicality and quantum phase transitions.

%

\acknowledgments K.C. Tan was supported by the NTU Presidential Postdoctoral Fellowship program funded by Nanyang Technological University.

\end{document}